\documentclass{ecai} 

\usepackage{latexsym}
\usepackage{amssymb}
\usepackage{amsmath}
\usepackage{amsthm}
\usepackage{booktabs}
\usepackage{enumitem}
\usepackage{graphicx}
\usepackage{color}

\usepackage{algorithm}
\usepackage{algpseudocode}

\newtheorem{theorem}{Theorem}
\newtheorem{lemma}[theorem]{Lemma}

\newenvironment{sketchofproof}{%
  \proof}{\endproof}

\newcommand{\BibTeX}{B\kern-.05em{\sc i\kern-.025em b}\kern-.08em\TeX}

\begin{document}


\begin{frontmatter}


\paperid{7430} 


\title{Multi-Agent Path Finding For Large Agents Is Intractable}


\author[1]{\fnms{Artem}~\snm{Agafonov}}
\author[2,1,3]{\fnms{Konstantin}~\snm{Yakovlev}}
\address{$^1$HSE University,~$^2$FRC CSC RAS,~$^3$AIRI}


\begin{abstract}
The multi-agent path finding (MAPF) problem asks to find a set of paths on a graph such that when synchronously following these paths the agents never encounter a conflict. In the most widespread MAPF formulation, the so-called Classical MAPF, the agents sizes are neglected and two types of conflicts are considered: occupying the same vertex or using the same edge at the same time step. Meanwhile in numerous practical applications, e.g. in robotics, taking into account the agents’ sizes is vital to ensure that the MAPF solutions can be safely executed. Introducing large agents yields an additional type of conflict arising when one agent follows an edge and its body overlaps with the body of another agent that is actually not using this same edge (e.g. staying still at some distinct vertex of the graph). Until now it was not clear how harder the problem gets when such conflicts are to be considered while planning. Specifically, it was known that Classical MAPF problem on an undirected graph can be solved in polynomial time, however no complete polynomial-time algorithm was presented to solve MAPF with large agents. In this paper we, for the first time, establish that the latter problem is NP-hard and, thus, if P$\neq$NP no polynomial algorithm for it can, unfortunately, be presented. Our proof is based on the prevalent in the field technique of reducing the seminal 3SAT problem (which is known to be an NP-complete problem) to the problem at hand. In particular, for an arbitrary 3SAT formula we procedurally construct a dedicated graph with specific start and goal vertices and show that the given 3SAT formula is satisfiable iff the corresponding path finding instance has a solution.

\end{abstract}

\end{frontmatter}


\section{Introduction}
Multi-agent path finding (MAPF)~\cite{stern2019multi} generalizes the task of finding a single path in a graph to the case of finding multiple paths with certain restrictions. In the basic MAPF formulation, the so-called Classical MAPF, we are given a graph and $n$ start and goal vertices, one for each agent. The time is discretized and at each time step an agent can either stay at its current vertex or move to an adjacent one following an edge. Thus the individual path for an agent may contain consecutive vertices, meaning that this agent is waiting for certain amount of time steps in them. The waiting actions are crucial in MAPF because they are often needed to eliminate \emph{conflicts}. In particular, two paths are said to contain an edge (vertex) conflict if they use the same edge (end up in the same vertex) at the same time step. The MAPF problem asks to find $n$ paths (one for each agent), s.t. any pair of them is conflict-free. 

\begin{figure}[t]
\centering
\includegraphics[width=0.9\linewidth]{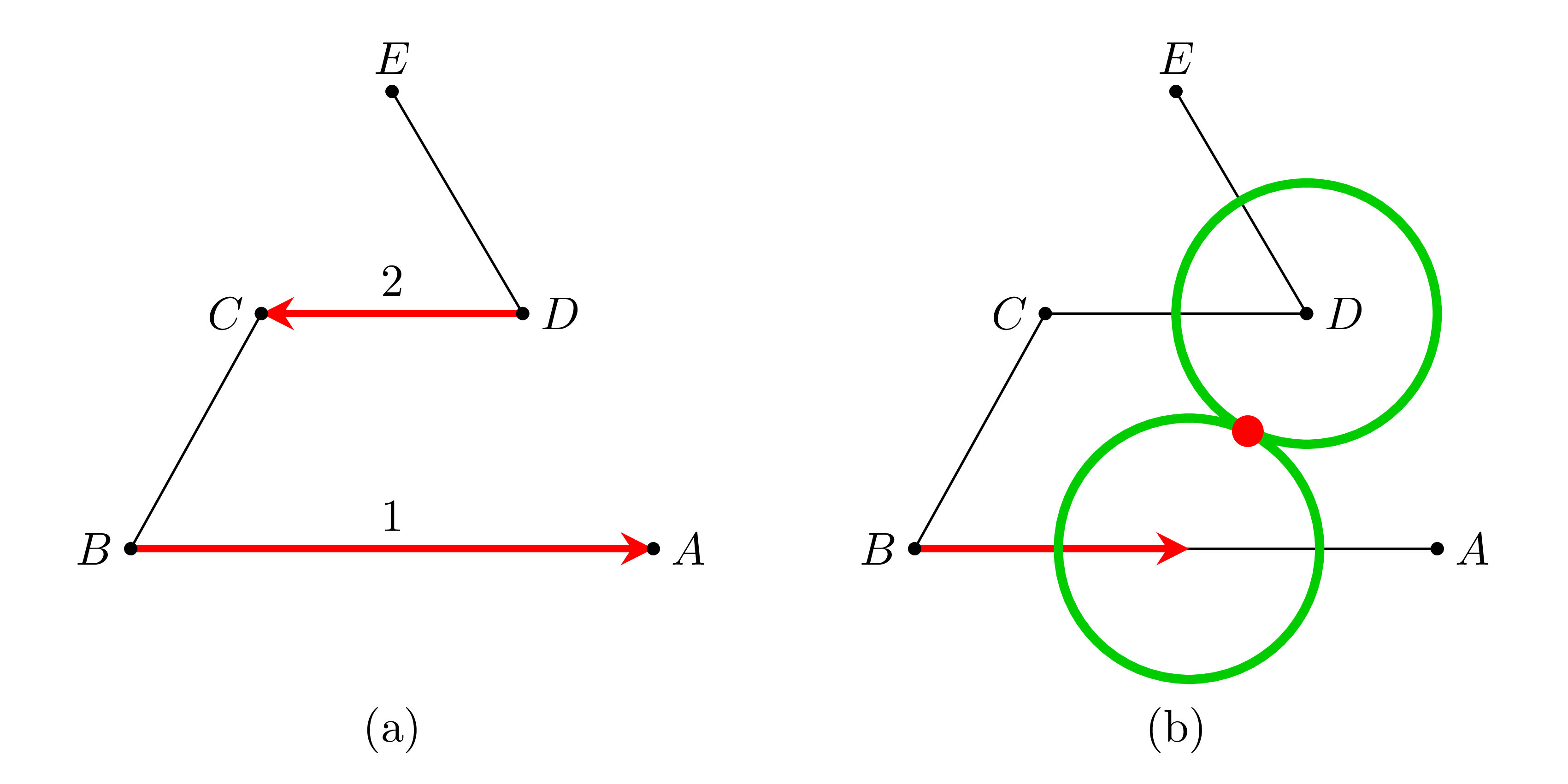}
\caption{MAPF vs. LA-MAPF.}
\label{fig:mapf_vs_la_mapf}
\end{figure}

The cost of a MAPF solution is typically measured as either the sum of costs of the individual solutions (SOC or flowitme) or the maximum over them -- the makespan. Here, the individual cost is the time step when the agent reaches its goal (and never moves away). It is known that solving MAPF optimally w.r.t. SOC or makespan is NP-Hard~\cite{surynek2010optimization}. Meanwhile, a decision variant of MAPF can be solved in polynomial time~\cite{kornhauser1984coordinating} (on an undirected graph) and, indeed, a plethora of fast, subotimal MAPF algorithms exists: Push and Rotate~\cite{de2014push}, PIBT~\cite{okumura2022priority}, LaCAM~\cite{okumura2023lacam} to name a few.

Meanwhile, directly transferring MAPF solutions to the real world is complicated as Classical MAPF relies on a range of simplifying assumptions that are not met in practice, e.g. body-less agents, uniform duration of moves, instantaneous acceleration/deceleration etc. To this end a less restrictive and more practically-oriented MAPF formulations were introduced and studied. One of such prominent formulation is the one introduced in~\cite{li2019multi} and named MAPF for Large Agents, or LA-MAPF for short. In this formulation the graph is naturally embedded into the Euclidean space (workspace) and it is assumed that an agent with body i) occupies a portion of the workspace while residing in the graph vertex; ii) sweeps a certain volume of space when moving along an edge which is explicitly associated with a segment of a straight line in the workspace. Consequently, the number of ways how the agents may get in conflict increases.

Consider for example a problem instance depicted in Fig.~\ref{fig:mapf_vs_la_mapf}. Here the first agent, $a_1$, has to move from $B$ to $A$ and the second one, $a_2$, from $D$ to $C$. In case of MAPF (Fig.~\ref{fig:mapf_vs_la_mapf}-a) the trivial solution: $a_1:B\rightarrow A, a_2: D\rightarrow C$ is collision-free. However the same solution for the case of LA-MAPF may be not valid (Fig.~\ref{fig:mapf_vs_la_mapf}-b). The problem is that no matter whether the disk-shaped agent $a_2$ stays at $D$ or moves from $D$ to $C$ the collision occurs with the same-sized agent $a_1$ traversing $BA$ edge. 
To solve the problem $a_2$ has to move to vertex $E$ first, making enough room for $a_1$ to move from $B$ to $A$. After this move $a_2$ can go back to $D$ and then to $C$.


Intuitively, LA-MAPF is harder to solve than MAPF, but until now the computational complexity of LA-MAPF was not formally established, contrary to MAPF whish is known to be solvable in polynomial time (on undirected graphs) since 1984~\cite{kornhauser1984coordinating}. In this work we fill this gap, by proving that the decision variant of LA-MAPF (i.e. giving YES/NO answer to the question ``if a solution for a LA-MAPF exists'') is NP-Hard. This implies that, if P $\neq$ NP, then there is no complete polynomial algorithm for solving LA-MAPF.



\section{Related works}
\label{section:related_works}

Two lines of research are especially relevant for our work: extending MAPF to handle agents that are not body-less (i.e. have shape and size) and analyzing the computational complexity of MAPF.



\paragraph{Extending MAPF To Handle Large Agents} It is known that Classical MAPF (with point agents) is solvable in polynomial time and one of the widespread polynomial MAPF algorithm with strong theoretical guarantees is Push and Rotate~\cite{de2014push}. In~\cite{dergachev2022towards} an adaptation of this algorithm for LA-MAPF was proposed that keeps the core logic of Push and Rotate but adds dedicated moves tailored to clear the way for the large agents in case of obstruction. The resultant algorithm runs in polynomial time, but, unfortunately, it is not complete, i.e. may fail to find a solution for a solvable LA-MAPF instance. 



Another widespread approach to solve Classical MAPF is the Conflict-Based Search (CBS)~\cite{sharon2015conflict}. This solver is originally tailored to obtain optimal MAPF solutions and the expected running time of the algorithm is estimated by an exponential function~\cite{gordon2021revisiting}, which limits its applicability in complex scenarios involving large number of agents. Moreover CBS may not correctly terminate on an usolvable MAPF instance. Still, it is widely used in MAPF community and adapted to various modifications of the problem.

The most relevant CBS adaptation to the problem we are interested in is presented in~\cite{li2019multi}, where CBS was specifically modified to handle large agents. Specifically, several new types of CBS constraints were introduced that increase the pruning power of CBS and positively influence its runtime.

In~\cite{yao2024layered} another adaptation of CBS to LA-MAPF problem is introduced. The key idea behind this method is to decompose the problem into smaller subproblems with a smaller number of agents. It does not compromise the guarantees and allows to find solutions faster in practice. Still, the complexity of this method remains exponential.

Overall, several adaptations of seminal MAPF methods to LA-MAPF setting are known but, none of such algorithms is both complete and runs in polynomial time. Thus, it is still not clear whether LA-MAPF can be solved in polynomial time (on an undirected graph) like regular MAPF, or not. 



\paragraph{MAPF Complexity}

A range of papers analyzing the computational complexity of different variants of Classical MAPF problems have been published so far. The key results are summarized in Table~\ref{table:complexity_analysis}. These works consider both variants of the problem: decision and optimization, ans they also investigate this problem under different assumptions about the graph, such as edge orientation and graph layout. P and NP results have been obtained, but all of this papers explore the problem under the assumption that agents do not have a geometric shape. Therefore, these papers do not directly answer our question, but they leave room for further research on the LA-MAPF problem.

\begin{table}[h]
\caption{Relevant results of complexity analysis for the MAPF problem.}~\label{table:complexity_analysis}
\centering
\begin{tabular}{p{0.2\linewidth}p{0.2\linewidth}p{0.2\linewidth}p{0.2\linewidth}} 
\toprule
\bf{Publication} & \bf{Graph} & \bf{Task type} & \bf{Complexity} \\
\midrule
Kornhauser. 1984~\cite{kornhauser1984coordinating} & Undirected & Decision & P \\
\midrule
Nebel. 2020~\cite{nebel2020computational} & Directed & Decision & NP-complete \\
\midrule
Nebel. 2024~\cite{nebel2024computational} & Directed, strongly-connected & Decision & P \\
\midrule
Yu. 2015~\cite{yu2015intractability} & Planar & Optimization & NP-complete \\
\midrule
Banfi et al. 2017~\cite{banfi2017intractability} & 2D grid & Optimization & NP-hard \\
\midrule
Tan et al. 2023~\cite{tan2023intractability} & Directed acyclic graph & Optimization & NP-complete \\
\bottomrule
\end{tabular}
\end{table}


\section{Problem Statement}
\label{section:problem_definition}



Let $G = (V, E)$ be an undirected graph that is embedded into the Euclidean plane $\mathbb{R}^2$. That is, each vertex $v \in V$ corresponds to a point on a plane with the coordinates $(v_x, v_y)$ and each edge $e=(u,v) \in E$ corresponds to a segment of a straight line connecting $(u_x, u_y)$ to $(v_x,v_y)$.

Let $\mathcal{A} = \{a_1, a_2, \ldots, a_k\}$ be a set of $k$ agents. Each agent is modeled as a disk of radius $r=const>0$, the agent's reference point is the center of the disk. We say that an agent is located at vertex $v \in V$ if its reference point has coordinates $(v_x, v_y)$. We say that an agent is moving along the edge $(u, v) \in E$ if its reference point is moving along the straight line segment connecting $(u_x, u_y)$ and $(v_x, v_y)$.

The time is discretised into the time steps, i.e. $T=[0, 1, ...)$. At each time stamp $t$ the joint configuration of the agents is an ordered sequence of distinct graph vertices, specifying where the agents are located: $\Pi_t=(v_1,...v_k)$. The location of the $i$th agent is denoted as $\Pi_t(a_i)=v_i$. Consider now two consecutive configurations  $\Pi_t$ and $\Pi_{t+1}$. These configurations are said to form a transition iff:
\begin{enumerate}
    \item $\exists! i\in[1,k]:~(\Pi_t(a_i),\Pi_{t+1}(a_i))\in E$
    \item $\forall j\neq i: ~\Pi_t(a_j)=\Pi_{t+1}(a_j)$
\end{enumerate}

In other words a transition occurs when exactly one agent moves along an edge while the other stay put.

Consider now a sequence $\Pi=(\Pi_1, \ \Pi_2, \ \ldots, \ \Pi_L)$. It is said to contain a vertex conflict, $\left<a_i, a_j, u, v, t\right>$, if $\exists i,j: \| \Pi_t(a_i),\Pi_t(a_j)\| < 2\cdot r$. I.e., the vertex conflict occurs when the agents bodies overlap. Similarly, the edge conflict is defined as $\left<a_i, a_j, u, v, v', t\right>$. Here, $a_i$ is the agent that changes its location from $u = \Pi_{t}(a_i)$ to $v = \Pi_{t+1}(a_i)$ and $a_j$ is the agent residing at $v'$ at $\Pi_{t}$ and $\Pi_{t+1}$, and the conflict occurs if the distance between the point $(v'_x, v'_y)$ and the segment connecting $(u_{x}, u_{y})$ and $(v_{x}, v_{y})$ is less than $2r$.

The LA-MAPF problem now is: given the start and the goal configurations find a sequence of transitions connecting them, s.t. it is conflict-free. Please note, that in this work we interested in getting any solution of LA-MAPF, rather than getting a solution that minimizes some cost objective (e.g. the number of transitions).

\section{Reduction Of 3-SAT To LA-MAPF}
\label{section:reduction}

Our ultimate goal is to demonstrate that LA-MAPF is an NP-Hard problem. We do this by following an established approach of reducing the seminal NP-complete problem, i.e. 3-SAT, to the problem at hand, i.e. LA-MAPF.

3-SAT is a problem that for a (special type of) Boolean formula asks whether there exists an assignment of the variables that renders this formula True. Specifically, the input formula is assumed to be a 3CNF, i.e. to be a conjunction of $m \geq 1$ clauses where each clause is the disjunction of exactly three literals: a variable or its negation. The total number of variables in the formula is denoted as $n \geq 1$.


Our reduction is a careful demonstration that for any 3-SAT instance, i.e. for arbitrary Boolean 3CNF formula, one can construct a graph and define the start and goal locations of the agents (as well as their radius), such that there exists a solution to the resulting LA-MAPF problem iff there is a satisfying assignment for the original 3-SAT formula. Additionally, we ensure that the size of the graph and the number of agents are polynomial in the number of variables and clauses of the initial 3-SAT instance.



The agents' radius in the reduction is set to be equal to the number of clauses in the original problem, i.e. $r = m$.

The reduction itself is made up of the three stages, each adding the so-called \textit{gadgets} -- subgraphs (that may contain start and goal locations for the agents) out of which the final LA-MAPF graph is constructed. Each vertex of each gadget has a specific coordinates on a plane. To better visualize them we divide the latter into the several disjoint zones as shown in Fig.~\ref{fig:zones}.

\begin{figure}[h]
\centering
\includegraphics[width=0.8\linewidth]{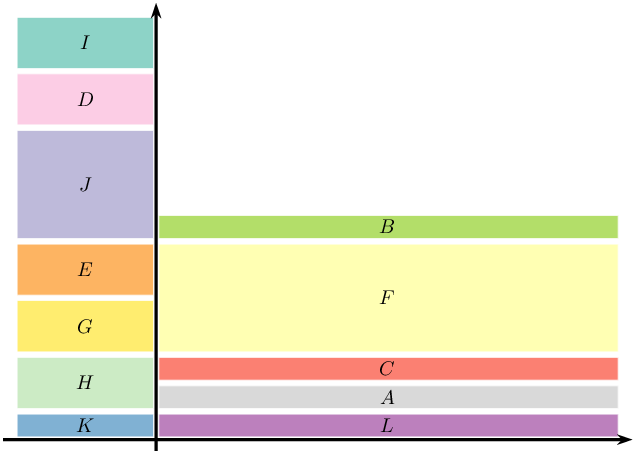}
\caption{The division of the plane for the reduction.}
\label{fig:zones}
\end{figure}

Please note, that for the sake of clarity, the locations of certain vertices in the subsequent figures will not exactly match their actual locations as per provided coordinates.

We will use a 3-SAT instance with $3$ variables and $1$ clause: $x_1 \vee \neg x_2 \vee x_3$, as an illustrative example (however, our reduction is indeed, applicable to any 3-SAT instance).






\subsection{Variables}

The first type of gadget we introduce is composed $3$ vertices ($A$, $B$, $C$) and an agent. The agent's starting and ending vertex is $A$. Even though this agent is already located at the final vertex, we will later show that it will need to move away from and return to this vertex in order to allow other agents to reach their destinations.

We populate this gadget $n$ times, one for each variable in the initial 3-SAT formula. We will refer to the agent associated with the $i$-th variable as the $i$-th v-agent. The vertices in the $i$-th gadget have the following coordinates:

$$\begin{cases}
    A_i = ((2m + 1) \cdot i, \ 1) \\
    B_i = ((2m + 1) \cdot i, \ 4m^2n + 2m^2 + 2mn + 5m - 1) \\
    C_i = ((2m + 1) \cdot i, \ 2m^2 + m)
\end{cases}$$

The gadgets of the first type used in the example are shown in Figure~\ref{fig:variable}. To improve clarity, the vertices of each gadget are not aligned along a vertical line.

\begin{figure}[h]
\centering
\includegraphics[width=0.8\linewidth]{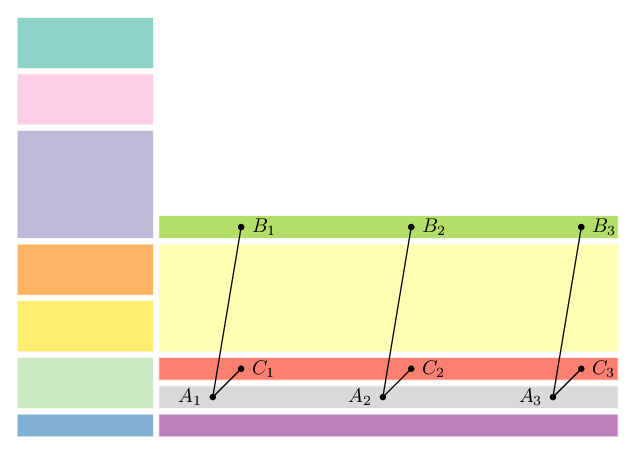}
\caption{The gadgets of first type.}
\label{fig:variable}
\end{figure}

We assume that the $i$-th v-agent, located at the $B_i$ vertex, corresponds to the true value of the $i$-th variable ($x_i = \text{True}$), and at the $C_i$ vertex corresponds to the false value ($x_i = \text{False}$).


\subsection{Clauses}

The second type of gadget has $8$ vertices ($D$, $E$, 3 vertices in the $F$ zone, $G$, $H$, $I$) and an agent. The agent's starting vertex is $D$ and ending vertex is $I$.

We populate this gadget $m$ times, one for each clause in the initial 3-SAT formula. We will refer to the agent associated with the $j$-th clause as the $j$-th c-agent. The vertices in the $j$-th gadget have the following coordinates:
$$\begin{cases}
    D_j = (0, \ (2m + 1) \cdot j + 4m^2n + 2m^2 + 4mn + 3m + n - 2) \\
    E_j = (0, \ j + 2m^2 n + 2m^2 + mn + 3m - 1) \\
    G_j = (0, \ j + 2m^2 n + 2m^2 + mn + 2m - 1) \\
    H_j = (0, \ (2m + 1) \cdot j) \\
    I_j = (0, \ (2m + 1) \cdot j + 4m^2n + 4m^2 + 4mn + 5m + n - 2)
\end{cases}$$
The coordinates of the remaining three vertices are calculated based on the literals in the $j$-th clause. If the $j$-th clause has a positive literal $x_i$, then we add a vertex to the graph with the following coordinates:
$$F_j^{x_i} = ((2m + 1) \cdot i, \ j + 2m^2 + 2m - 1)$$
If the $j$-th clause has a negative literal $\neg x_i$, then we add a vertex to the graph with the following coordinates:
$$F_j^{\neg x_i} = ((2m + 1) \cdot i, \ j + 4m^2n + 2m^2 + 2mn + 3m - 1)$$

The gadget of the second type that is used in the example are shown in Figure~\ref{fig:clause}. To improve clarity, the vertices $D_1$, $E_1$, $G_1$, $H_1$ and $I_1$ are not aligned along a vertical line.

\begin{figure}[h]
\centering
\includegraphics[width=0.8\linewidth]{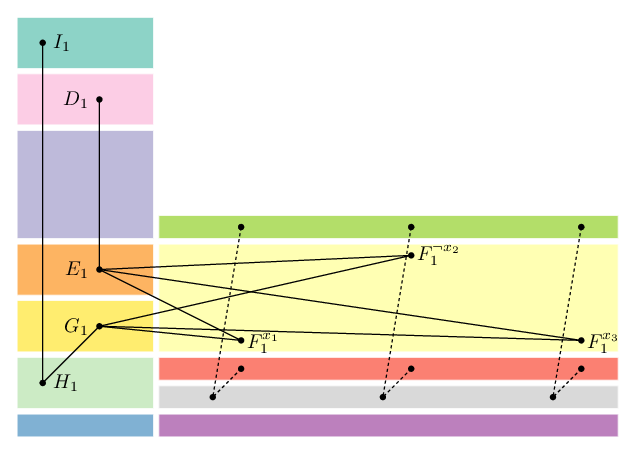}
\caption{The subgraph of the clause-agent.}
\label{fig:clause}
\end{figure}

\begin{lemma}
\label{lemma:variable_blocks_clause}
If there are agents in the $C_i$ and $F_j^{x_i}$ vertices (similarly for the $B_i$ and $F_j^{\neg x_i}$ vertices), then there is a conflict between them.
\end{lemma}
\begin{proof}
The distance between the $C_i$ and $F_j^{x_i}$ vertices is:
$$\sqrt{(x_{F_j^{x_i}} - x_{C_i})^2 + (y_{F_j^{x_i}} - y_{C_i})^2} =$$
$$= \sqrt{(((2m + 1) \cdot i) - ((2m + 1) \cdot i))^2 +  (y_{F_j^{x_i}} - y_{C_i})^2} =$$
$$= |(j + 2m^2 + 2m - 1) - (2m^2 + m)| =$$
$$= j + m - 1 < 2r$$
Similarly, for the $B_i$ and $F_j^{\neg x_i}$ vertices:
$$\sqrt{(x_{B_i} - x_{F_j^{\neg x_i}})^2 + (y_{B_i} - y_{F_j^{\neg x_i}})^2} =$$
$$= \sqrt{(((2m + 1) \cdot i) - ((2m + 1) \cdot i))^2 + (y_{B_i} - y_{F_j^{\neg x_i}})^2} =$$
$$= |y_{B_i} - y_{F_j^{\neg x_i}}| = 2m - j < 2r$$
Since the distance between the vertices is less than $2r$, the simultaneous presence of agents at these vertices will lead to a conflict between them.
\end{proof}

Lemma~\ref{lemma:variable_blocks_clause} states that if the $i$-th v-agent is located at the $C_i$ vertex, then the $j$-th c-agent cannot move to the $F_j^{x_i}$ vertex, as this would cause a collision between the two agents. A similar property holds for the vertices $B_i$ and $F_j^{\neg x_i}$. This restriction creates a connection between the two problems. An inappropriate assignment of variables leads to the unsatisfiability of a clause. Similarly, an inappropriate placement of v-agents may prevent a c-agent from reaching its destination.


\subsection{Blocking}

The final gadget we discuss has $2n + 1$ vertices ($J_1$-$J_n$, $K$, $L_1$-$L_n$) and $n + 1$ agents. The starting and ending vertices for the agents are $J_1$-$J_n$ and $K$. Although these agents are already at their final vertices, they will need to move away from and return to these vertices in order for c-agents to reach their destinations. To accomplish this, they must force v-agents to leave their own final vertices, as we mentioned earlier.

We will refer to the agent in this gadget as the b-agent. The vertices in this gadget have the following coordinates:

$$\begin{cases}
    J_i = (0, \ (2m + 1) \cdot i + 4m^2n + 2m^2 + 2mn + 3m - 2) \\
    K = (0, \ 0) \\
    L_i = ((2m + 1) \cdot i, \ 0)
\end{cases}$$

The gadget of the third type used in the example is shown in Figure~\ref{fig:block}.

\begin{figure}[h]
\centering
\includegraphics[width=0.8\linewidth]{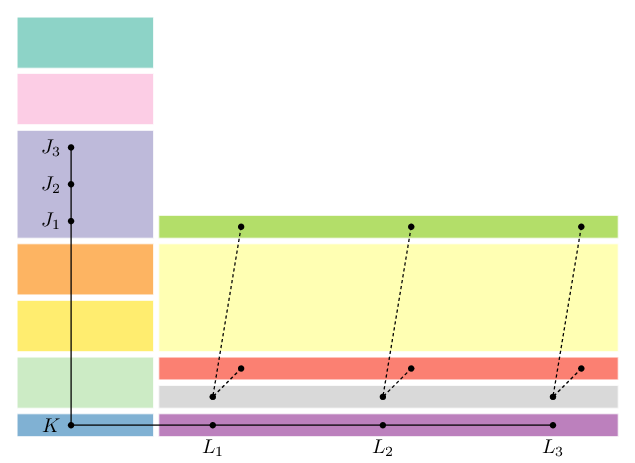}
\caption{The subgraph of the block-agents.}
\label{fig:block}
\end{figure}

\begin{lemma}
\label{lemma:block_blocks_variable}
If there are agents in the $A_i$ and $L_i$ vertices, then there is a conflict between them.
\end{lemma}
\begin{proof}
The distance between the $A_i$ and $L_i$ vertices is:
$$\sqrt{(x_{A_i} - x_{L_i})^2 + (y_{A_i} - y_{L_i})^2} =$$
$$= \sqrt{(((2m + 1) \cdot i) - ((2m + 1) \cdot i))^2 + (1 - 0)^2} =$$
$$= 1 < 2r$$
Since the distance between the vertices is less than $2r$, the simultaneous presence of agents at these vertices will lead to a conflict between them.
\end{proof}

Lemma~\ref{lemma:block_blocks_variable} states that if the b-agent is located at the $L_i$ vertex, then the $i$-th v-agent cannot move to the $A_i$ vertex, as this would result in a collision between these two agents. This restriction prevents the $i$-th v-agent from moving between $B_i$ and $C_i$ vertices. It other words, the values of all variables remain fixed whenever the b-agents are located at the $L$ vertices.


\section{Correctness of reduction}~\label{section:correctness}

In this section, we will demonstrate the equivalence between the original 3-SAT problem and the LA-MAPF problem that we have constructed. We will prove that if one of the problems has a solution, then so does another problem. The following subsections will present algorithms for constructing the solution to one problem from the given solution to the other.


\subsection{3-SAT to LA-MAPF}

Given a solution to the 3-SAT problem, which is an assignment of variables that makes a formula satisfiable, we can use Algorithm~\ref{alg:3sat_to_lamapf} to create a conflict-free plan that moves all agents to their destinations in the constructed reduction.

\begin{algorithm}
\caption{The construction of a reduction solution based on the 3-SAT solution.}\label{alg:3sat_to_lamapf}
\textbf{Input:} $\{x_i \in \{ True, \ False\}: \ i = 1, \ \ldots, \ n\}$ - Solution of 3-SAT \\
\textbf{Result:} $\{(u_1 \rightarrow v_1), \ \ldots, \ (u_N \rightarrow v_N)\}$ - Solution of LA-MAPF
\begin{algorithmic}[1]
\State Solution $\gets \{\}$
\For{$i \gets 1$ \textbf{to} $n$}
    \If{$x_i$ is True}
        \State Solution.Append($A_i \rightarrow B_i$)
    \Else
        \State Solution.Append($A_i \rightarrow C_i$)
    \EndIf
\EndFor
\State Solution.Append($K \rightarrow L_1 \rightarrow \ldots \rightarrow L_n$)
\For{$i \gets 1$ \textbf{to} $n - 1$}
    \State Solution.Append($J_i \rightarrow \ldots \rightarrow J_1 \rightarrow K \rightarrow L_1 \ldots \rightarrow L_{n - i}$)
\EndFor
\State Solution.Append($J_n \rightarrow \ldots \rightarrow J_1 \rightarrow K$)
\For{$j \gets 1$ \textbf{to} $m$}
    \State Solution.Append($D_j \rightarrow E_j$)
    \For {$(id, \ neg)$ \textbf{in} arguments of $j$-th clause}
        \If{$neg$ is True}
            \State $to \gets F_j^{\neg x_{id}}$
        \Else
            \State $to \gets F_j^{x_{id}}$
        \EndIf
        \If{$x_{id}$ is not $neg$}
            \State Solution.Append($E_j \rightarrow to \rightarrow G_j$)
            \State \textbf{break}
        \EndIf
    \EndFor
    \State Solution.Append($G_j \rightarrow H_j$)
\EndFor
\For{$j \gets m$ \textbf{to} $1$}
    \State Solution.Append($H_j \rightarrow I_j$)
\EndFor
\State Solution.Append($K \rightarrow J_1 \rightarrow \ldots \rightarrow J_n$)
\For{$i \gets n-1$ \textbf{to} $1$}
    \State Solution.Append($L_{n - i} \rightarrow \ldots \rightarrow L_1 \rightarrow K \rightarrow J_1 \ldots \rightarrow J_i$)
\EndFor
\State Solution.Append($L_n \rightarrow \ldots \rightarrow L_1 \rightarrow K$)
\For{$i \gets n$ \textbf{to} $1$}
    \If{$x_i$ is True}
        \State Solution.Append($B_i \rightarrow A_i$)
    \Else
        \State Solution.Append($C_i \rightarrow A_i$)
    \EndIf
\EndFor
\State \Return Solution
\end{algorithmic}
\end{algorithm}

The solution constructed by Algorithm~\ref{alg:3sat_to_lamapf} consists of six steps. Figure~\ref{fig:3sat_to_lamapf} illustrates these steps using the example of a 3-SAT formula with $3$ variables $x_1 = x_2 = True$, $x_3 = False$, and $1$ clause: $x_1 \vee \neg x_2 \vee x_3$.

\begin{figure}[h]
\centering
\includegraphics[width=0.8\linewidth]{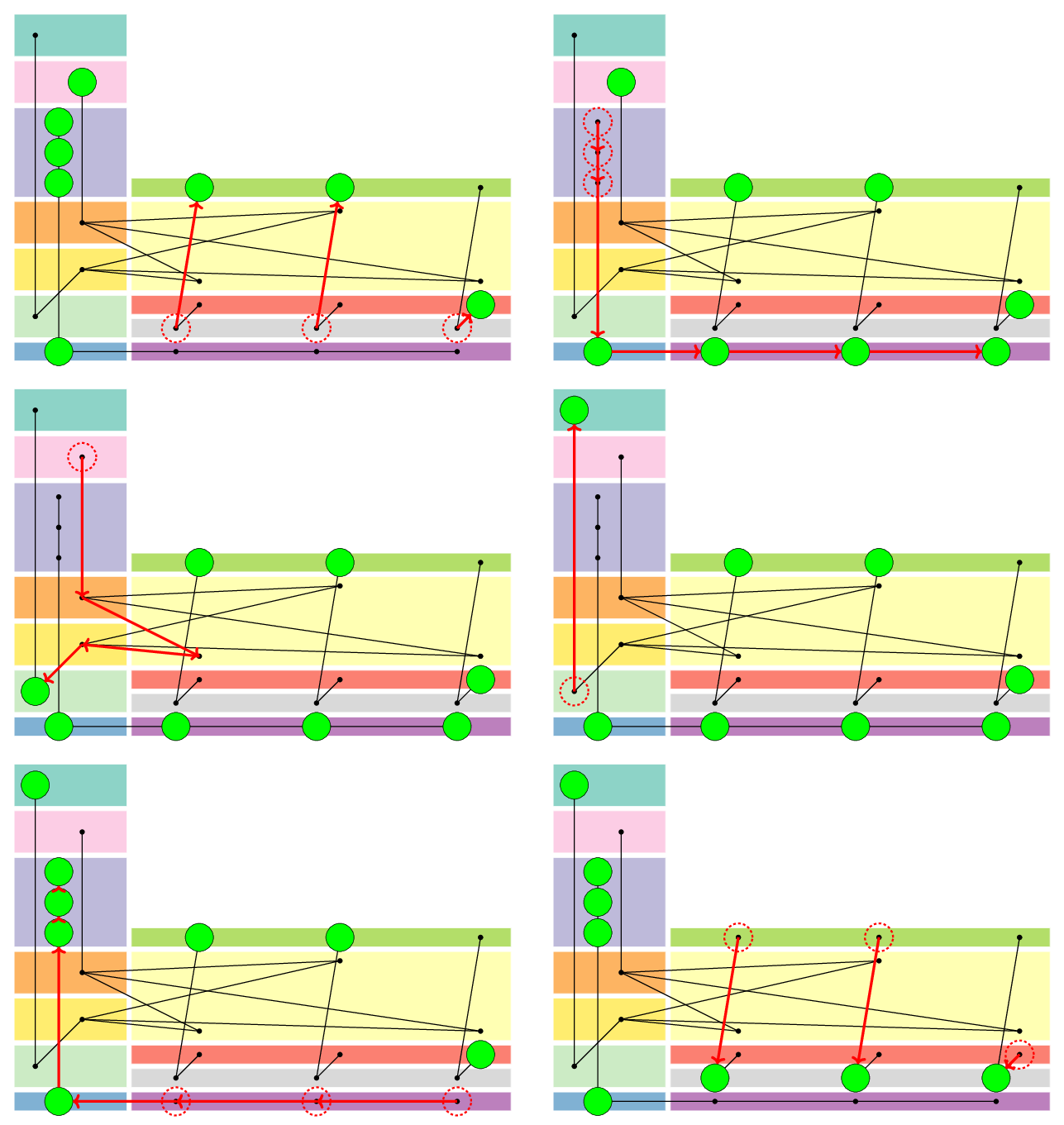}
\caption{The steps of the solution constructed by Algorithm~\ref{alg:3sat_to_lamapf}.}
\label{fig:3sat_to_lamapf}
\end{figure}

In lines 2-8, the v-agents are moved to the vertices that correspond to the given solution. This means that if the $i$-th variable has a true value, then the $i$-th v-agent is moved to the $B_i$ vertex, and to the $C_i$ vertex otherwise. In the example, the v-agents are moved to vertices $B_1$, $B_2$ and $C_3$, respectively.

In lines 9-13, the b-agents are moved to the $L$ vertices. This action empties $J$ vertices and makes it possible to move the c-agents.

In lines 14-28, the c-agents are moved from the $D$ vertex to the $H$ vertex. The cycle in line 16 chooses the literal of the clause that makes it satisfied. In line 23, the c-agent is moved through the selected vertex. In the example, the c-agent moves through vertex $F_1^{x_1}$, as the true value of the first variable satisfies this clause.

In lines 29-31, the c-agents are moved to their respective destinations.

Finally, in lines 32-43, the actions taken in lines 2-13 are reversed. This action returns the b-agents and v-agents back to their destinations.

Lemma~\ref{lemma:3sat_to_lamapf} proves that Algorithm~\ref{alg:3sat_to_lamapf} constructs conflict-free paths.

\begin{lemma}
\label{lemma:3sat_to_lamapf}
The Algorithm~\ref{alg:3sat_to_lamapf} provides a correct solution to the constructed reduction based on the given 3-SAT solution.
\end{lemma}
\begin{sketchofproof}
In lines 2-8, the moves are made only by the v-agents. Therefore, there are no agents in the $F$ and $L$ zones. This allows all v-agents to move freely to the vertices in the $B$ and $C$ zones. After that, the $A$ zone becomes free of agents.

In lines 9-13, the moves are made by the b-agents. Since there are no agents in the $A$, $E$, $G$, and $H$ zones, they are able to move to the vertices of the $K$ and $L$ zones. After that, the $J$ zone becomes empty.

In lines 14-28, the c-agents are moved through the $D$, $E$, $F$, $G$ and $H$ zones. Note that in line 23, the choice of a conflict-free vertex $F_j^x$ is always possible. This is because in the initial 3-SAT solution, there was a literal that satisfied this clause, so the c-agent will be able to move through $F_j^x$ vertex corresponding to that literal.

The moves in lines 29-31 are conflict-free.

In lines 32-43, the moves of the v-agents and the b-agents are the same as those in lines 2-13, but reversed. Therefore they are also conflict-free.

It means that the algorithm creates a plan for the agents to move towards their destinations without interfering with each other. This proves the correctness of this algorithm.
\end{sketchofproof}


\subsection{LA-MAPF to 3-SAT}

Given a solution to the reduction, which is a plan that moves all agents to their destinations without conflicts, we can use Algorithm~\ref{alg:lamapf_to_3sat} to assign variables in such way that satisfies the original 3-SAT formula.

\begin{algorithm}
\caption{The construction of a 3-SAT solution based on the solution of the reduction.}\label{alg:lamapf_to_3sat}
\textbf{Input:} $\{(u_1 \rightarrow v_1), \ \ldots, \ (u_N \rightarrow v_N)\}$ - Solution of LA-MAPF \\
\textbf{Result:} $\{x_i \in \{ True, \ False\}: \ \ i = 1, \ \ldots, \ n\}$ - Solution of 3-SAT
\begin{algorithmic}[1]
\State State $\gets$ state where all agents are at their destinations
\For{$i \gets N$ \textbf{to} $1$}
    \State State.Move($v_i$, $u_i$)
    \If{State.IsOccupied($G_1$)}
        \State \textbf{break}
    \EndIf
\EndFor
\State Solution $\gets \{\}$
\For{$i \gets 1$ \textbf{to} $n$}
    \If{State.IsOccupied($B_i$)}
        \State Solution $\gets$ Solution $\cup \ \{x_i = True\}$
    \Else
        \State Solution $\gets$ Solution $\cup \ \{x_i = False\}$
    \EndIf
\EndFor
\State \Return Solution
\end{algorithmic}
\end{algorithm}

In the first line, the LA-MAPF state is initialized with the final state, where all agents are in their final positions. In lines 2-7, the agents' moves are reversed until the $G_1$ vertex is occupied. Finally, in lines 9-15, variables are set: if the $B_i$ vertex is occupied, then $x_i$ is set to true, otherwise is set to false.

Lemma~\ref{lemma:lamapf_to_3sat} proves that Algorithm~\ref{alg:lamapf_to_3sat} sets variables in such a way that the initial 3-SAT formula is satisfied.

\begin{lemma}
\label{lemma:lamapf_to_3sat}
The Algorithm~\ref{alg:lamapf_to_3sat} provides a correct solution for the initial 3-SAT problem, based on the given solution to the reduction.
\end{lemma}
\begin{proof}
Consider the last time the first c-agent was in the $G_1$ vertex. Because in the end it must be in the $I_1$ vertex, the next move must lead it to the $H_1$ vertex.

Note that there is no b-agent in the $J$ zone. Assume that it is there. Then, it will not be able to move out of the $J$ zone, as the first c-agent will block its movement through the edge between the $J_1$ and $K$ vertices. This means that the first c-agent will not be able to reach the $I_1$ vertex, as this would lead to a collision with the b-agent in the $J$ zone.

This means that all b-agents are located in the $K$ and $L$ zones, and therefore, all v-agents must be located in the $B$ and $C$ zones, since they cannot be located in the $A$ zone according to Lemma~\ref{lemma:block_blocks_variable}.

Consider the next timestamp, when the first c-agent will move to the $I_1$ vertex. Since it has stayed in the $H_1$ vertex, no b-agent can leave the $L$ zone. Consequently, none of the v-agents have changed their position. Also, this move means that no c-agent is in the $D$ or $E$ zones, so every c-agent was able to make a conflict-free move to a vertex in the $F$ zone, based on the current positions of the v-agents.

Assume that there is a clause in the initial 3-SAT problem that is not satisfied by the constructed variable assignment. This implies that if we move the corresponding c-agent to any vertex of its gadget in the $F$ zone, it would lead to a conflict with the corresponding v-agent according to Lemma~\ref{lemma:variable_blocks_clause}. However, this would contradict the last statement in the previous paragraph. Therefore, all clauses must be satisfied, and the algorithm has constructed a correct solution to the initial 3-SAT problem. This proves the correctness of the algorithm.
\end{proof}


\section{On The Intractability Of LA-MAPF Problem}~\label{section:intractability}

We have shown that by using the proposed reduction, we can create a LA-MAPF problem that is solvable iff the initial 3-SAT formula has a solution. This means that the LA-MAPF problem is an NP-hard problem.

\begin{theorem}
\label{theorem:np_hardness}
The LA-MAPF problem is an NP-hard problem.
\end{theorem}
\begin{proof}
The number of vertices in the graph that was constructed is:
$$n \cdot 3 + m \cdot 8 + (2n + 1) = 5n + 8m + 1 \in \text{Poly}(n, m)$$
The number of agents involved in the reduction is:
$$n + m + (n + 1) = 2n + m + 1 \in \text{Poly}(n, m)$$
Therefore, the size of the reduction depends on the size of the original 3-SAT formula in a polynomial way. Lemma~\ref{lemma:3sat_to_lamapf} states that if the 3-SAT formula has a solution, then the reduction also has the solution. Lemma~\ref{lemma:lamapf_to_3sat} states that if the reduction has a solution, then the initial 3-SAT formula has a solution as well. This means that we have constructed a polynomial reduction from 3-SAT to LA-MAPF. Since 3-SAT is an NP-complete problem, this proofs that LA-MAPF is NP-hard.
\end{proof}


\section{Discussion}
\label{section:discussion}

In this work we have assumed that the agents are confined to a graph embedded in $\mathbb{R}^2$ and each agent is modeled as a disk. Generally, the dimensionality of the workspace may be larger and the sizes and shapes of the agents may vary. Still, by showing that at least one variant of LA-MAPF is NP-Hard, i.e. the one considered in the paper, one can infer that the generic formulation is NP-Hard as well.  


Another assumption that we have adopted is that only one agent can move at a time. Meanwhile, synchronous simultaneous moves of several agents may be allowed in MAPF as long as they do not result in collision. Our reduction is valid for simultaneous moves assumption as long as a transition from one joint configuration to a consecutive joint configuration can be decomposed into the sequence of atomic, i.e. single-agent, moves (as we are not focused on the cost of the solution but rather on establishing the presence of the solution). It worth to note here, that this is not always the case. For example of a joint simultaneous move that can no be safely decomposed into the series of single-agent moves consider a ``closed-chain'' move, when at a time step a group of agents move in such way that the first agent moves to the vertex that was occupied by the second agent, the second agent moves to the vertex that was occupied by the third agent and so on, while the last agent moves to the vertex occupied by the first one. Estimating LA-MAPF tractability in this case is a possible direction of future work.



\section{Conclusions And Future Work}
\label{section:conclusion}

In this paper, we have investigated the decision version of the LA-MAPF problem, which is a natural extension of the MAPF problem where size of the agents is to be taken into account. So far the computational complexity of this problem had not yet been determined, in contrast to regular MAPF. To this end  we have proposed a reduction from the seminal 3-SAT problem to the LA-MAPF problem in order to demonstrate that the latter is intractable, i.e. NP-hard. This means that (unlike regular MAPF) LA-MAPF can not solved in polynomial time (even on undirected graphs) unless P=NP.
Our research establishes an important theoretical result and justifies the development of new heuristics for existing algorithms and faster exponential-time algorithms for practical applications.


The main reason why LA-MAPF is hard to solve is the need to consider closeness of the agents in the metric space to which the MAPF graph is embebedded. A prosperous direction of future research is to focus on specific graph embeddings and layouts that may make the problem simpler (while still adequately reflecting the constraints of the real-world MAPF applications). Examples of such graphs are planar graphs, where edges can only intersect at vertices, or grids. The proposed reduction method is generic and not tailored to these cases, so there may still exist a polynomial-time algorithm that could solve the problem by taking advantage of a specific graph layout.

Another fruitful direction for extension of this work, as mentioned previously, is considering the setting where several agents are allowed to move simultaneously. 
Finally, we have demonstrated that LA-MAPF is NP-hard, but we still do not know whether it belongs to NP. If one shows that LA-MAPF is in NP, this would also imply that it is NP-complete. One possible way to prove this would be to demonstrate that for any instance of LA-MAPF, a solution of polynomial length can be found.

Overall our study not only establishes an important milestone in MAPF research but opens a range of avenues for future research.



\begin{ack}
\end{ack}



\bibliography{mybibfile}

\end{document}